\documentclass{acm_proc_article-sp}
\usepackage{graphicx}
\usepackage{amssymb}
\usepackage{epstopdf}
\usepackage[mathcal]{euscript}
\usepackage{amsmath}
\usepackage{enumerate}
\usepackage{cite}

\newcommand{\A}{{\mathcal A}}
\newcommand{\B}{{\mathcal B}}
\newcommand{\C}{{\mathcal C}}
\newcommand{\E}{{\mathcal E}}
\newcommand{\F}{{\mathcal F}}
\newcommand{\HH}{{\mathcal H}}
\newcommand{\Q}{{\mathcal Q}}

\newcommand{\vu}{{\sf u}}
\newcommand{\vv}{{\sf v}}
\newcommand{\vx}{{\sf x}}
\newcommand{\vy}{{\sf y}}

\newcommand{\bbZ}{{\mathbb Z}}

\newtheorem{corollary}{Corollary}
\newtheorem{definition}{Definition}
\newtheorem{lemma}{Lemma}
\newtheorem{theorem}{Theorem}
\newtheorem{problem}{Problem}
\newtheorem{example}{Example}

\begin{document}

\title{Optimal Memoryless Encoding for Low Power Off-Chip Data Buses}

\numberofauthors{3}

\author{
\alignauthor Yeow Meng Chee\titlenote{Y. M. Chee is also with Card View Pte. Ltd.,
41 Science Park Road, \#04-08A The Gemini, Singapore Science Park II, Singapore 117610.} \\
	\affaddr{School of Physical \& Mathematical Sciences} \\
	\affaddr{Nanyang Technological University} \\
	\affaddr{Singapore 637616} \\
	\email{ymchee@alumni.uwaterloo.ca}
\alignauthor Charles J. Colbourn \\
	\affaddr{Department of Computer Science \& Engineering} \\
	\affaddr{Arizona State University} \\
	\affaddr{Tempe, Arizona 85287-8809} \\
	\email{charles.colbourn@asu.edu}
\alignauthor Alan C. H. Ling \\
	\affaddr{Department of Computer Science} \\
	\affaddr{University of Vermont} \\
	\affaddr{Burlington, Vermont 05405} \\
	\email{aling@emba.uvm.edu}
}

\maketitle

\begin{abstract}
Off-chip buses account for a significant portion of the total system power consumed in
embedded systems. Bus encoding schemes have been proposed
to minimize power dissipation, but none has been demonstrated
to be optimal with respect to any measure. In this paper, we give the first
provably optimal and explicit (polynomial-time constructible) families of memoryless
codes for minimizing bit transitions in off-chip buses. Our results imply that having access
to a clock does not make a memoryless encoding scheme that minimizes bit transitions more powerful.
\end{abstract}

\section{Introduction}

Energy efficiency is an important product quality
characteristic. For mobile applications, such as handheld and wireless devices,
it not only impacts the usability and acceptance directly, but also
affects reliability and packaging cost of the product. Consequently, design
techniques for minimizing system power consumption are important
for achieving high product quality.

Power-efficient design requires the reduction of power dissipation throughout the design,
during all stages of the design process, subject to constraints on system
performance and quality of service. In CMOS circuits, most power is dissipated as
dynamic power for charging and discharging of internal node capacitances.
Thus, researchers have investigated techniques for minimizing the number of
transitions inside the circuits. The power dissipation at the input/output (I/O) pads
of an integrated circuit (IC) is even higher because off-chip buses have switching
capacitances that are orders of magnitude greater than those internal to a chip.
The power dissipated at the I/O pads of an IC ranges from 10\% to 80\% of the
total power dissipation with a typical value of 50\% for circuits optimized
for low power \cite{StanBurleson:1995b}. The concern of this paper is with
low power encoding for off-chip buses.

Bus encoding is used to reduce power dissipated on the bus lines. It has been
shown to be an effective technique for power reduction. Compared to on-chip
buses of deep submicron (DSM) circuits, the ratio of inter-wire capacitances
(or crosstalk coupling) to substrate capacitance is much lower in off-chip buses.
The energy model for off-chip buses is thus simpler. We follow the model
discussed by Catthoor {\em et al.} \cite{Catthooretal:1998}. For a bus with
$k$ wires, each one having metal interconnect capacitance of $C$, and under
a supply voltage $V_{dd}$, the total energy consumed by the bus for a
computation of $N$ cycles can be estimated as
\begin{equation}
\label{energyequation}
E_{bus} = N \cdot k \cdot C \cdot V_{dd}^2.
\end{equation}
The average power can then be obtained by multiplying $E_{bus}$ by $f$, the frequency
at which the bus operates. However, the derivation of (\ref{energyequation}) assumes that for each
cycle, all $k$ bus wires exhibit a transition that causes the corresponding
capacitance to switch. In general, not all the wires switch; hence a cycle-based model
is more accurate, as in:
\begin{equation}
\label{newenergyequation}
E_{bus} = \left( \sum_{i=1}^N k_i\right) \cdot C \cdot V_{dd}^2,
\end{equation}
where $k_i$ is the number of wires switching at cycle $i$.

The only freedom in (\ref{newenergyequation}) for reducing energy consumption
is $k_i$, since $N$ depends on the application running on the processor, and $C$
and $V_{dd}$ depend on the technology. Therefore, bus encoding techniques for
reducing energy consumption have focused on reducing $k_i$, the bit switching
activity of the off-chip bus. However, no explicit encoding scheme with provable
optimality (with respect to any formal measure) is known. In this paper, we give
the first explicit and provably optimal memoryless encoding schemes (with respect to rate
and maximum energy consumed per cycle), both stateless and
stateful, for off-chip buses.

\section{Mathematical Preliminaries}

\subsection{Codes}

The {\em Hamming $n$-space} is the set $\HH(n)=\{0,1\}^n$ endowed with the
{\em Hamming distance} $d_H$ defined as follows:
for $\vu,\vv\in\HH(n)$, $d_H(\vu,\vv)$ is the number of positions where $\vu$
and $\vv$ differ. The {\em Hamming weight} (for short, {\em weight}) of a vector
$\vu\in\HH(n)$ is the number of positions in $\vu$ with non-zero value, and is
denoted $w_H(\vu)$. The $i$th component of $\vu$ is denoted $\vu_i$. The
{\em support} of a vector $\vu\in\HH(n)$, denoted supp$(\vu)$, is the set
$\{i:\vu_i=1\}$.

Any subset of $\HH(n)$ is called a {\em code of length $n$}. A
{\em constant weight code of length $n$ and weight $w$} is any subset of
$\HH(n,w)=\{\vu\in\HH(n):w_H(\vu)=w\}$. The elements of a code are called
{\em codewords}. Let $\C\subseteq\HH(n)$ be a code. The {\em size} of $\C$
is $|\C|$, the number of codewords in the code. The {\em diameter} of $\C$ is
diam$(\C)=\max_{\vu,\vv\in\C} d_H(\vu,\vv)$. Given two codes
$\C_1,\C_2\subseteq\HH(n)$, the {\em cross diameter} of $\C_1$ and $\C_2$
is crossdiam$(\C_1,\C_2)=\max_{\vu\in\C_1,\vv\in\C_2} d_H(\vu,\vv)$.
By definition, crossdiam$(\C,\C)={\rm diam}(\C)$.

\subsection{Set Systems}

For integers $i<j$, the set $\{i,i+1,\ldots,j\}$ is abbreviated as $[i,j]$. Moreover,
we also write $[j]$ for $[1,j]$. For a finite set $X$ and $k\leq |X|$, we define
$2^X=\{A:A\subseteq X\}$ and ${X\choose k}=\{A\in2^X:|A|=k\}$. The ring
$\bbZ/n\bbZ$ is denoted $\bbZ_n$.

A {\em set system of order $n$} is a pair $(X,\A)$, where $X$ is a finite set
of $n$ {\em points} and $\A\subseteq 2^X$. The elements of $\A$ are called
{\em blocks}. A set system is said to be $k$-{\em uniform} if $\A\subseteq{X\choose k}$.

Let $([n],\A)$ be a set system. The {\em incidence vector} of a block $A\in\A$ is the
vector $\iota(A)\in\HH(n)$, where
\begin{equation*}
\iota(A)_i =  
\begin{cases}
1,& \text{if $i\in A$; and} \\
0,& \text{otherwise.}
\end{cases}
\end{equation*}
There is a natural correspondence between the Hamming $n$-space and the
{\em complete} set system of order $n$, $([n],2^{[n]})$: the positions of vectors
in $\HH(n)$ correspond to points in $[n]$, a vector $\vu\in\HH(n)$ corresponds to
the block supp$(\vu)$, and $d_H(\vu,\vv)=|{\rm supp}(\vu)\Delta{\rm supp}(\vv)|$.
From this, it follows that there is a bijection between the set of all codes of length
$n$ and the set of all set systems of order $n$. There is also a bijection between
the set of all constant weight codes of length $n$ and weight $w$ and the
set of all $w$-uniform set systems of order $n$. So we may speak of
{\em the set system of a code}, or {\em the code of a set system}.

\section{Problem Formulation}

Let $s\geq 1$ and $n\geq k$. An $s$-{\em state} $n$-{\em bit encoding scheme
for a source} $S\subseteq\HH(k)$ is a triple
$\E=\langle \C,E,D\rangle$, where
\begin{enumerate}
\item $\C$ is a code of length $n$,
\item $E:S\times\bbZ_s\rightarrow \C$
is an injective map called an {\em encoding function}, and
\item $D:\C\times\bbZ_s\rightarrow S$
is a surjective map called a {\em decoding function},
\end{enumerate}
such that $D(E(\vu))=\vu$ for all $\vu\in S$.
The encoding function induces a subscode $\C_i\subseteq\C$ for each $i\in\bbZ_s$, defined by
\begin{equation*}
\C_i = \{\vv\in\C : \text{$\vv=E(\vu,i)$ for some $\vu\in S$}\}.
\end{equation*}
Conversely, $\{{\C_i}, i\in\bbZ_s\}$ uniquely identifies the class of functions that
$E$ (and hence $D$) can belong to:
\begin{align*}
E \in
& \{ f:S\times\bbZ_s\rightarrow \C :  \\
& ~~~~~~~~~~~~~\text{$f(\cdot,i):S\rightarrow \C_i$
is injective for all $i\in\bbZ_s$}\}.
\end{align*}

\begin{figure}
\centering
\includegraphics[width=2.5in]{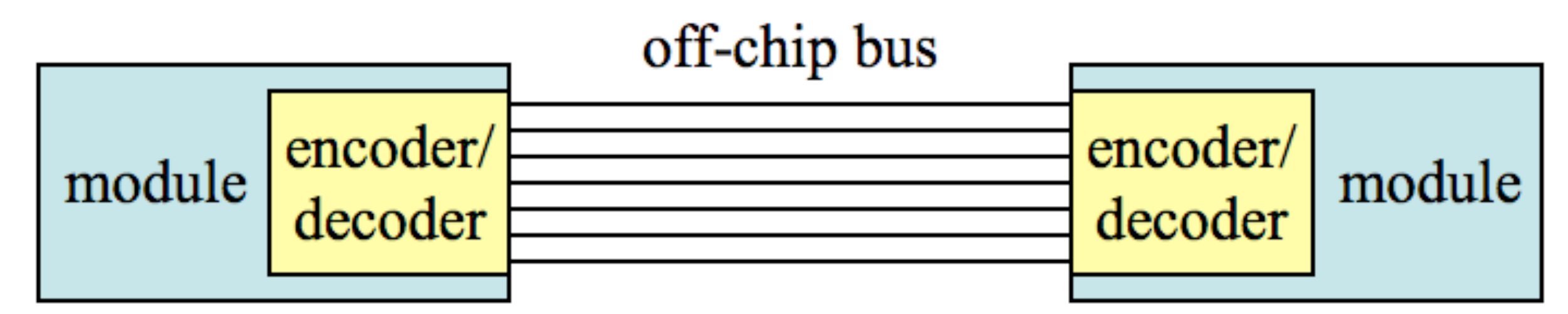}
\caption{Non-adaptive encoding for off-chip bus}
\label{nonadaptive}
\end{figure}
Encoding and decoding modules that implement  $E$ and $D$,
respectively, are inserted at the interface of the bus (see Fig. \ref{nonadaptive}).
Suppose that $\vu,\vv\in S$ are two words to be communicated
across the bus at steps $t$ and $t+1$ of a computation. In the absence of a bus
encoding scheme, the number of bit switchings in this computation cycle is
$|\{i:\vu_i\not=\vv_i\}|$. This quantity is precisely $d_H(\vu,\vv)$, which can be as high as $k$.
If an $s$-state $n$-bit encoding scheme is used, then $\vu'=E(\vu,i)$
and $\vv'=E(\vv,j)$ are communicated instead in steps $t$ and $t+1$
of the computation, where $i=t\pmod{s}$ and $j=t+1\pmod{s}$.
The resulting number of bit switchings in this computation cycle is
therefore $d_H(\vu',\vv')$, which is bounded above by
crossdiam$(\C_i,\C_j)$.

To ensure that the encoding scheme guarantees performance under as general
a condition as possible, we adopt a worst-case analysis model. 
Indeed for a given performance guarantee $\delta\leq k$, we require that
crossdiam$(\C_i,\C_{i+1})\leq \delta$ for all $i\in\bbZ_s$. We call such a
family of codes $(\C_i)_{i\in\bbZ_s}$ an {\em $(n,\delta)_s$-low power} (LP) {\em code}.

A 1-state encoding scheme is known as a {\em stateless} encoding scheme since
it does not need to know the state (step number) of the computation in order to encode
the word on the bus. In this case, an $(n,\delta)_1$-LP code is a code $\C$
such that diam$(\C)\leq\delta$. Any $s$-state encoding scheme with $s\geq 2$ is called
{\em stateful}. While a stateless encoding scheme does not need access to a clock, a stateful encoding scheme would require access to a clock to know the particular computation cycle.

\begin{figure}
\centering
\includegraphics[width=2.5in]{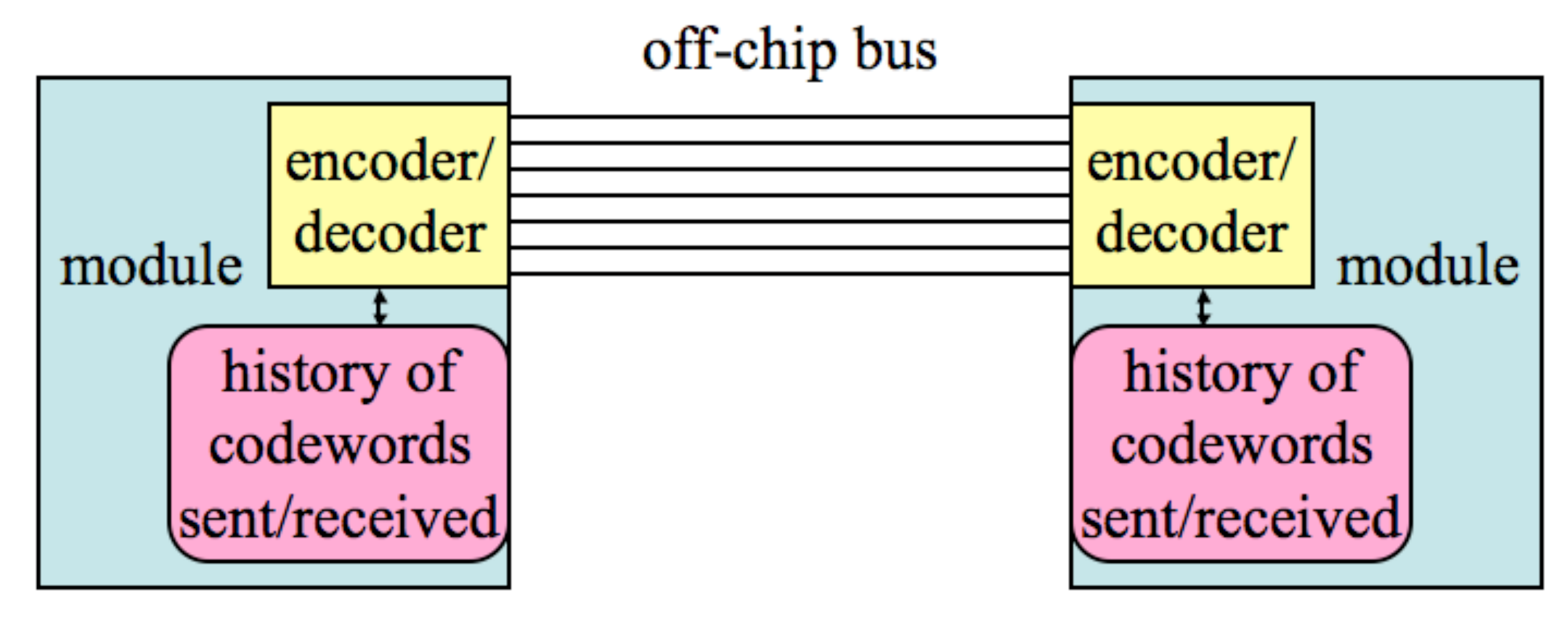}
\caption{Adaptive encoding for off-chip bus}
\label{adaptive}
\end{figure}
The encoding schemes that we consider are {\em non-adaptive}:
The choice of codeword to transmit across the bus in step $t$ does not depend
on codewords that have been transmitted in steps $i$, $i<t$. A diagram
depicting {\em adaptive} encoding is given in Fig. \ref{adaptive}.
Both non-adaptive and adaptive encoding schemes for low-power data buses have been considered (see, for example,
\cite{Beninietal:1997,ChengPedram:2001,ChengPedram:2002,Fornaciarietal:2000,Lindkvistetal:2004,MurgaiFujita:1999,PatelMarkov:2004,Shinetal:1998,StanBurleson:1995a,StanBurleson:1995b,
Subrahmanyaetal:2004}).
Adaptive encoding schemes  have better rates.
However non-adaptive encoding schemes are often simpler to implement
since they do not need to have a continuously changing data model, and do not
require memory (to track history of transmitted codewords). Thus non-adaptive encoding schemes are also known as {\em memoryless} encoding schemes.
The focus of this paper is on memoryless encoding schemes.

The {\em wire expansion} of $\E$, an $n$-bit encoding scheme for a source $S\subseteq\HH(k)$,
is the quantity $\alpha(\E)=n/k$. Wire expansion should be minimized since
it affects the system area. Hence, given $\delta\leq k$, we want to find an
$(n,\delta)_s$-LP code $(\C_i)_{i\in\bbZ_s}$ with the smallest $n$ such that
$|\C_i|\geq 2^k$, for all $i\in\bbZ_s$. In this way we arrive at the following equivalent problem:

\begin{problem}
Given $n$, $\delta$, and $s$, find an $(n,\delta)_s$-LP code
$(\C_i)_{i\in\bbZ_s}$ maximizing $\min_{i\in\bbZ_s}|\C_i|$.
\end{problem}

\noindent Let
\begin{equation*}
F_s(n,\delta) =
 \max_\text{$(n,\delta)_s$-LP code $(\C_i)_{i\in\bbZ_s}$} \ 
\min_{i\in\bbZ_s} \ \ |\C_i| .
\end{equation*}
An $(n,\delta)_s$-LP code $(\C_i)_{i\in\bbZ_s}$ such that
$\min_{i\in\bbZ_s} |\C_i| =F_s(n,\delta)$ is called
{\em optimal}.

\begin{lemma}
\label{statereduction}
$F_s(n,\delta)\leq F_2(n,\delta)$ for all $s\geq 2$.
\end{lemma}

\begin{proof}
This follows from the property that if $(C_i)_{i\in\bbZ_s}$ is an $(n,\delta)_s$-LP code, then for
any $A\subseteq\bbZ_s$ such that $|A|\geq 2$, $(\C_i)_{i\in A}$ is also an $(n,\delta)_{|A|}$-LP
code.
\end{proof}

Lemma \ref{statereduction} shows that there is no advantage in using $s$-state
encoding schemes with $s>2$. Hence, we restrict our attention to $s=1$ (the stateless
case) and $s=2$ (the stateful case) for the rest of this paper. 

\begin{problem}
Given $n$ and $\delta$, construct an optimal
$(n,\delta)_1$-LP code.
\end{problem}

\begin{problem}
Given $n$ and $\delta$, construct an optimal
$(n,\delta)_2$-LP code.
\end{problem}

For $\vu,\vv\in\HH(n,w)$, we have $d_H(\vu,\vv)\leq \min\{2w,n\}$. Therefore one
way to constrain the number of bit switchings is to limit the code to contain codewords of
low weight. This suggests the consideration of constant weight codes.

\begin{problem}
Given $n$ and $\delta$, construct an optimal
$(n,\delta)_1$-LP code of constant weight $w$.
\end{problem}

\begin{problem}
Given $n$ and $\delta$, construct an optimal
$(n,\delta)_2$-LP code of constant weight $w$.
\end{problem}

In subsequent sections, we present complete solutions to Problems 2, 3, and 4.
Problem 5 is solved for $n$ sufficiently large.

\section{Optimal Stateless Low Power Codes}
\label{stateless}

By definition, $([n],\A)$ is the set system of an
$(n,\delta)_1$-LP code if and only if $|A_1\Delta A_2|\leq \delta$ for all $A_1,A_2\in\A$.

The $(n,\delta)_1$-LP codes are known as {\em anticodes of
length $n$ and maximum distance $\delta$} and were introduced in the coding theory literature
by Farrell \cite{Farrell:1970} in 1970. However, set systems $([n],\A)$ satisfying
$|A_1\Delta A_2|\leq \delta$ for all $A_1,A_2\in\A$ were studied earlier by
Kleitman \cite{Kleitman:1966}, who obtained a complete solution to the problem of determining
the maximum number of blocks in such a set system.

\begin{theorem}[Kleitman]
\label{Kleitman}
Let $N(n,\delta)$ denote the maximum number of blocks in a set system $([n],\A)$
satisfying $|A_1\Delta A_2|\leq \delta$ for all $A_1,A_2\in\A$. Then
\begin{equation*}
 N(n,\delta) = 
 \begin{cases}
\sum_{i=0}^{\delta/2} {n\choose i},&\text{if $\delta\equiv 0\pmod{2}$;} \\
\sum_{i=0}^{(\delta-1)/2} {n\choose i}+{n-1\choose (\delta-1)/2},&\text{if $\delta\equiv 1\pmod{2}$.}
\end{cases}
\end{equation*}
The extremal set systems are given by
\begin{equation*}
\A = 
\begin{cases}
\bigcup_{i=0}^{\delta/2}{[n]\choose i}, &\text{if $\delta\equiv 0\pmod{2}$;} \\
\bigcup_{i=0}^{(\delta-1)/2}{[n]\choose i}\cup \\ \{A\cup\{x\}:A\in{[n]\setminus\{x\}\choose(\delta-1)/2}\}, &
\text{if $\delta\equiv 1\pmod{2}$,}
\end{cases}
\end{equation*}
where $x$ is any fixed element of $[n]$.
\end{theorem}

The explicit description of optimal $(n,\delta)_1$-LP
codes provided by Theorem \ref{Kleitman}
shows that such codes can be constructed in time polynomial in the size of the code.
This solves Problem 2.

\subsection{Restriction to Constant Weight}

We now address constant weight $(n,\delta)_1$-LP codes.
The diameter of any constant weight $(n,\delta)_1$-LP code is even, so we need
only consider $\delta\equiv 0\pmod{2}$. In this subsection, we show the equivalence of constant
weight $(n,\delta)_1$-LP codes and {\em intersecting families}.

\begin{definition}
A set system $(X,\A)$ is a {\em $t$-intersecting family} if $|A_1\cap A_2|\geq t$ holds for all
$A_1,A_2\in\A$.
\end{definition}

\begin{lemma}
\label{cwcmemoryless}
The code $\C$ of a set system $([n],\A)$ is an $(n,\delta)_1$-LP code of constant weight $w$
if and only if $([n],\A)$ is a $w$-uniform $\left(w-\delta/2\right)$-intersecting family.
\end{lemma}

\begin{proof}
Since $([n],\A)$ is $\left(w-\delta/2\right)$-intersecting,
\begin{eqnarray*}
|A_1\cap A_2| \geq w-\delta/2 & \Leftrightarrow & |A_1\Delta A_2| \leq \delta \\
& \Leftrightarrow & d_H(\iota(A_1),\iota(A_2)) \leq \delta,
\end{eqnarray*}
for all $A_1,A_2\in\A$.
Hence, $\C$ has diameter at most $\delta$.
\end{proof}

Let $M(n,k,t)$ be the maximum number of blocks in a $k$-uniform $t$-intersecting
family of order $n$. The determination of $M(n,k,t)$ and the structure of maximum $k$-uniform
$t$-intersecting families was initiated by Erd\H{o}s {\em et al.} \cite{Erdosetal:1961},
who proved the celebrated Erd\H{o}s-Ko-Rado Theorem. This result has been
improved subsequently by Frankl \cite{Frankl:1978}, Wilson \cite{Wilson:1984}, and
Ahlswede and Khachatrian \cite{AhlswedeKhachatrian:1997}, culminating in the
following.

\begin{theorem}[Complete Intersection Theorem]\hfill
\label{CIT}
Let $1\leq t\leq k\leq n$ and
define $\F(i)=\{F\in{[n]\choose k} : |F\cap [t+2i]|\geq t+i\}$. Let
$r = \arg \max_{i\in[0,(n-t)/2]} |\F(i)|$. Then $M(n,k,t)=|\F(r)|$ and
$\F(r)$ is a maximum $k$-uniform $t$-intersecting family of order $n$.
\end{theorem}

Lemma \ref{cwcmemoryless} together with Theorem \ref{CIT} completely determines
the optimal $(n,\delta)_1$-LP codes of constant weight $w$, solving Problem 4.

\section{Optimal Stateful Low Power Codes}

Two set systems $([n],\A)$ and $([n],\B)$ form a {\em $\delta$-pair} if $|A\Delta B|\leq\delta$
for all $A\in\A$ and $B\in\B$. Let $([n],\A)$ and $([n],\B)$
be the set systems of $\C_0$ and $\C_1$, respectively. Then from the definition, $([n],\A)$ and $([n],\B)$ form a $\delta$-pair if and only if $(\C_0,\C_1)$ is an
$(n,\delta)_2$-LP code.

Let $0\leq r\leq n$.
A {\em quasi-sphere of dimension $n$ and radius $r$} is a set $\Q$ of the form
$\Q=\{A\subseteq [n] : 0\leq |A|\leq r\}\cup \bar{\Q}$, where $\bar{\Q}$ consists of the first $L$
elements of $\{A\subseteq [n]: |A|=r+1\}$ in lexicographic order, for some $L\geq 0$.
Given $0\leq N\leq 2^n$, there exists a unique quasi-sphere $\Q$ of dimension $n$ such that
$|\Q|=N$ (see \cite{Katona:1975} for a proof).

We want to determine the maximum of $\min\{|\A|,|\B|\}$, where $([n],\A)$
and $([n],\B)$ form a $\delta$-pair. A related result has been obtained by Ahlswede and
Katona \cite{AhlswedeKatona:1977}.

\begin{theorem}[Ahlswede and Katona]
\label{AhlswedeKatona}
Let $1\leq N\leq 2^n$, $1\leq\delta\leq n$, and let $([n],\A)$ and $([n],\B)$
be a $\delta$-pair with $|\A|=N$. Then $\max |\B|$ is achieved if $\A$ is a
quasi-sphere and $\B=\{B : \text{$|B\Delta A|\leq \delta$ for all $A\in\A$}\}$.
\end{theorem}

We assume without loss of generality that $\delta < n$, since otherwise we may take
$\A=\B=2^{[n]}$.
Suppose that $([n],\A)$ and $([n],\B)$ form a $\delta$-pair, with $|\A|=N$, such that
$\B$ has the maximum number of blocks. Write
$N=\sum_{i=0}^{r}{n\choose i}+m$, $0\leq m\leq {n\choose r+1}$, for some $r\leq n$.
Theorem \ref{AhlswedeKatona} asserts that $\A$ is the quasi-sphere of 
dimension $n$ and radius $r$, with $|\A|=N$,
and $\B=\{B\subseteq[n]: \text{$|B\Delta A|\leq \delta$ for all $A\in\A$}\}$.
Suppose that $B\in\B$ and $|B|=k$. Then there exists $A\in([n]\setminus B)$ with $|A|=\delta-k+1$
whenever $\delta<n$. Since $|A\Delta B|=\delta+1$, we have $A\not\in\A$. This is possible
only if $\delta-k+1> r$, or $k< \delta-r+1$.
It follows that no $k$-subset of $[n]$, $k\geq \delta-r+1$, is in $\B$. Hence,
$|\B|\leq \sum_{i=0}^{\delta-r}{n\choose i}$.

Now consider the case that $\delta\equiv 0$ (mod 2).
If $r\geq\delta/2$, then $|\B|\leq\sum_{i=0}^{\delta/2}{n\choose i}\leq N(n,\delta)$. If $r<\delta/2$, then
$|\A|\leq\sum_{i=0}^{\delta/2-1}{n\choose i}+{n\choose \delta/2}=\sum_{i=0}^{\delta/2}{n\choose i}
\leq N(n,\delta)$. Hence, we have $\min\{|\A|,|\B|\} \leq N(n,\delta)$.

Next, consider the case that $\delta\equiv 1$ (mod 2). If $r\geq (\delta+1)/2$, then
$|\B|\leq\sum_{i=0}^{(\delta-1)/2}{n\choose i}\leq N(n,\delta)$. If $r\leq (\delta-3)/2$, then
$|\A|\leq\sum_{i=0}^{(\delta-3)/2}{n\choose i}+{n\choose (\delta-1)/2}\leq
\sum_{i=0}^{(\delta-3)/2}{n\choose i}+{n-1\choose(\delta-1)/2}\leq N(n,\delta)$. Hence,
$\min\{|\A|,|\B|\} \leq N(n,\delta)$, when $r\not=(\delta-1)/2$.

We deal with the remaining case in which $\delta\equiv 1$ (mod 2) and $r=(\delta-1)/2$.
If $m\leq {n-1\choose r}$, then $|\A|\leq\sum_{i=0}^{(\delta-1)/2}{n\choose i}+{n-1\choose{(\delta-1)/2}}
\leq N(n,\delta)$. If $m>{n-1\choose r}$, then by definition of a quasi-sphere,
$\A$ contains the sets $\{1\}\cup R$, where $R\in{[2,n]\choose r}$. We claim that $1\in B$ for every
$B\in\B$ such that $|B|=r+1$. Suppose not, let $R\subseteq [n]\setminus B$ such that
$|R|=r$. Then $A=\{1\}\cup R\in\A$, and $|A\Delta B|=2r+2>\delta$, giving a contradiction.
It follows that all subsets of size $r+1$ in $\B$ must contain the element $1$, and hence
the number of subsets of size $r+1$ in $\B$ is at most ${n-1\choose r}$. Hence,
$|\B|\leq\sum_{i=0}^{(\delta-1)/2}{n\choose i}+{n-1\choose (\delta-1)/2}\leq N(n,\delta)$.

This establishes the following:

\begin{theorem}
For any $\delta$-pair $([n],\A)$ and $([n],\B)$, 
\begin{align*}
& \min\{|\A|,|\B|\} \leq \\
& \begin{cases}
\sum_{i=0}^{\delta/2} {n\choose i},&\text{if $\delta\equiv 0\pmod{2}$;} \\
\sum_{i=0}^{(\delta-1)/2} {n\choose i}+{n-1\choose (\delta-1)/2},&\text{if $\delta\equiv 1\pmod{2}$,}
\end{cases}
\end{align*}
with equality if
\begin{equation*}
 \A = \B =
 \begin{cases}
\bigcup_{i=0}^{\delta/2}{[n]\choose i}, &\text{if $\delta\equiv 0\pmod{2}$} \\
\bigcup_{i=0}^{(\delta-1)/2}{[n]\choose i}\cup \\ \{A\cup\{x\}:A\in{[n]\setminus\{x\}\choose(\delta-1)/2}\}, &
\text{if $\delta\equiv 1\pmod{2}$,}
\end{cases}
\end{equation*}
where $x$ is any fixed element of $[n]$.
\end{theorem}

\begin{corollary}
\label{nohelp}
$F_1(n,\delta)=F_2(n,\delta)$; 
the size of an optimal $(n,\delta)_2$-LP code is the same as the size of an optimal $(n,\delta)_1$-LP
code.
\end{corollary}

This solves Problem 3.

\subsection{Restriction to Constant Weight}

Two set systems $([n],\A)$ and $([n],\B)$ are called {\em cross-wise $t$-intersecting}
if $|A\cap B|\geq t$ for all $A\in\A$ and $B\in\B$.

\begin{lemma}
Let $([n],\A)$ and $([n],\B)$ be the set systems of codes $\C_0$ and $\C_1$, respectively.
Then $(\C_0,\C_1)$ is an $(n,\delta)_2$-LP code of constant weight $w$ if and only if $([n],\A)$
and $([n],\B)$ are $w$-uniform and cross-wise $\left(w-\delta/2\right)$-intersecting.
\end{lemma}

\begin{proof}
Follows from the observation that
$([n],\A)$ and $([n],\B)$ are cross-wise $\left(w-\delta/2\right)$-intersecting if and only if
for all $A\in\A$ and $B\in\B$,
\begin{eqnarray*}
|A\cap B| \geq w-\delta/2 & \Leftrightarrow & |A\Delta B| \leq \delta \\
& \Leftrightarrow & d_H(\iota(A),\iota(B))\leq\delta.
\end{eqnarray*}
Hence, $(\C_0,\C_1)$ is an $(n,\delta)_2$-LP code.
\end{proof}

Let $([n],\A)$ and $([n],\B)$ be two set systems that are $w$-uniform cross-wise $t$-intersecting, where
$t=w-\delta/2$. Suppose that $A_1,A_2\in\A$.
Let $B\in\B$, $U\subseteq A_1\cap A_2\cap B$, and $u=|U|$. Then
$B$ must contain a further $t-u$ points from each of $A_1\setminus U$ and $A_2\setminus U$.
The remaining $w-(2t-u)$ points in $B$ are from the $n-(2t-u)$ points not already contained
in $B$.
The total possible number of such blocks $B$ is then
\begin{equation}
\label{blocks}
{w-u\choose t-u}^2{n-(2t-u)\choose w-(2t-u)}.
\end{equation}
If $u\leq t-1$, then for $n$ large enough, (\ref{blocks}) is less than
$M(n,k,t)={n-t\choose w-t}$, and hence $|\B|\leq M(n,k,t)$.

On the other hand, if $u\geq t$, then $|A_1\cap A_2|\geq t$, which implies $\A$ is
$t$-intersecting. So we also have $|\A|\leq M(n,k,t)$.
This gives the following result.

\begin{theorem}
Suppose $([n],\A)$ and $([n],\B)$ are $w$-uniform cross-wise $t$-intersecting set systems.
Then for large enough $n$,  $\min\{|\A|,|\B|\}\leq {n-t\choose w-t}$, with
equality if $\A=\B=\{A\in{[n]\choose w}: T\subseteq A\}$, where $T$ is any fixed $t$-subset of $[n]$.
\end{theorem}

\begin{corollary}
For $n$ large enough, the size of an optimal $(n,\delta)_2$-LP code of constant weight $w$
is the same as the size of an optimal $(n,\delta)_1$-LP code of constant weight $w$.
\end{corollary}

This solves Problem 5 for $n$ large enough. A solution for all $n$ seems out of reach at the moment.

\section{Implementation}

The encoding schemes introduced can be easily implemented with a look-up table,
since our codes are explicit and constructible in polynomial time. However, it is
also possible to encode and decode algorithmically, removing the need to store a look-up table of size
$2^k$ when the souce is $\HH(k)$.
We illustrate this with an encoding/decoding algorithm for the optimal
$(n,\delta)_1$-LP code from Section \ref{stateless}.
We assume that the source is $\HH(k)$.

First we define ranking and unranking algorithms.
Given a set of objects $X$, {\em rank} and {\em unrank} are functions
$f:X\rightarrow\{0,\ldots,|X|-1\}$ and $g:\{0,\ldots,|X|-1\}\rightarrow X$, such that
$f$ and $g$ are bijections satisfying $g(f(x))=x$ for all $x\in X$. Computing $f$ is 
{\em ranking}, and computing $g$ is {\em unranking}.
Efficient ranking and unranking
algorithms for $\HH(n,w)$ (equivalent to $w$-subsets of an $n$-set)
are well known (see, for example, \cite{NijenhuisWilf:1978}). A simple and efficient ranking and
unranking algorithm for $\HH(n,w)$ using the {\em co-lexicographic} ordering is described below.
\begin{center}
\begin{tabular}{l}
{\bf \boldmath rank$(w,\vx)$} $\{$ \\
\ \ \ \ $\{t_1,\ldots,t_w\}={\rm supp}(\vx)$; \\
\ \ \ \ $r=0$; \\
\ \ \ \ for $i=1$ to $w$ \\
\ \ \ \ \ \ \ \ $r = r+{t_i-1\choose k+1-i}$; \\
\ \ \ \ return $r$; \\
$\}$ \\
\\
\\
\\
\end{tabular}
\begin{tabular}{l}
{\bf \boldmath unrank$(n,w,r)$} $\{$\\
\ \ \ \ $x=n$; \\
\ \ \ \ for $i=1$ to $w$ $\{$ \\
\ \ \ \ \ \ \ \ while ${x\choose w+1-i} > r$ \\
\ \ \ \ \ \ \ \ \ \ \ \ $x = x-1$; \\
\ \ \ \ \ \ \ \ $t_i=x+1$; \\
\ \ \ \ \ \ \ \ $r = r-{x\choose k+1-i}$; \\
\ \ \ \ $\}$ \\
\ \ \ \ return $\iota(\{t_1,\ldots,t_w\})$; \\
$\}$
\end{tabular}
\end{center}
We now give a high-level overview of the method.
We interpret $\vx\in\HH(k)$ as an integer in the interval $[0,2^k-1]$ in the natural way.
Let $N_i=\sum_{j=0}^i {n\choose j}-1$, for $0\leq i\leq n$. By convention, $N_{-1}=-1$.
An $\vx\in[N_{t-1}+1,N_t]$ is encoded as an element of $\HH(n,t)$ by
unranking $\vx-N_{t-1}-1$. To decode a received $\vy\in\HH(n)$, simply rank $\vy$.

More precisely,
to encode a source $\vx\in\HH(k)$ to a codeword $\vy$, the encoder performs the following steps:
\begin{description}
\item[Step 1:] Find $t$ such that $N_{t-1}<\vx\leq N_t$.
\item[Step 2:] $m = \vx-N_{t-1}-1$.
\item[Step 3:] $\vy = {\rm unrank}(n,t,m)$.
\end{description}
The decoding algorithm is even simpler. To decode $\vy\in\HH(n)$
to a source vector $\vx$, the decoder performs the following steps:
\begin{description}
\item[Step 1:] Compute $w=w_H(\vy)$.
\item[Step 2:] $\vx={\rm rank}(w,\vy)$.
\end{description}

An optimal constant weight $(n,\delta)_1$-LP code can be similarly implemented.

\section{Not all Limited Weight Codes are Equal}

In \cite{StanBurleson:1995a}, Stan and Burleson introduced the following bus encoding scheme.
Define an {\em $m$-limited weight code} ($m$-LWC) of length $n$ to be a
code $\C\subseteq \cup_{w=0}^m \HH(n,w)$.
Suppose that the number of source states to be transmitted across a bus is $2^k$ and
that the source states are to be encoded with a code of length $n$. Let
$m$ be the smallest integer such that
\begin{equation*}
\sum_{i=0}^m {n\choose m}  \geq 2^k.
\end{equation*}
Stan and Burleson claimed that an $m$-LWC code of length $n$,
comprising as codewords all elements in
$\cup_{w=0}^{m-1} \HH(n,w)$ and the remaining $2^k-\sum_{i=0}^m {n\choose m}$ codewords
from $\HH(n,m)$, is 
\begin{quote}
``{\em optimal in the sense that any other code with the same
length cannot have better statistical properties for low power.}''
\end{quote}
This statement is true since we have shown in Section \ref{stateless}
that an optimal $(n,\delta)_1$-LP code is an $m$-LWC,
but we must exercise caution when choosing the
$m$-LWC, since not every $m$-LWC is optimal, as can be seen in the example below.

\begin{example}
Let $k=5$ and $n=6$. We give two $3$-LWCs of length $n$ and size $2^k$
with different diameters. Each of the two $3$-LWCs of length $n$ and size $2^k$ contains
all $22$ elements in $\cup_{w=0}^2\HH(6,w)$.
The first $3$-LWC containing the $10$ additional codewords
\begin{equation*}
\begin{array}{c c c c c}
000111 & 001011 & 010011 & 100011 & 001110 \\
111000 & 110100 & 101100 & 011100 & 110001
\end{array}
\end{equation*}
has diameter six, and
the second $3$-LWC containing the $10$ additional codewords
\begin{equation*}
\begin{array}{c c c c c}
100011 & 100101 & 100110 & 101001 & 101010 \\
101100 & 110001 & 110010 & 110100 & 111000
\end{array}
\end{equation*}
has diameter five.
\end{example}

Our results in Section \ref{stateless} give $m$-LWCs that are optimal, for every length.

\section{Conclusion}

Past research on encoding schemes for low-power buses has largely been experimental; no optimal codes were known for any
measure, in any model. This paper  constructs codes that are provably optimal,
starting with the simplest model, that of an off-chip data bus.
In so doing, we obtain the first explicit and
provably optimal memoryless encoding scheme that minimizes bit
transitions for off-chip data buses. Our approach is combinatorial and the
codes obtained are explicit and polynomial-time constructible.
We also show that having access to a clock (or alternatively, knowing the computation cycle)
does not help in achieving more efficient encoding.

We are currently extending this work to thermal-aware models and models
at the DSM level, where inter-wire
capacitances are significant, and crosstalks must be avoided.
\\

\section*{Acknowledgments}

The authors are grateful to Rudolf Ahlswede for helpful pointers.

\bibliographystyle{abbrv}
\bibliography{/Users/ymchee/Documents/Personal/Bibliography/mybibliography}

\end{document}